\documentclass[11pt]{article}
\usepackage[utf8]{inputenc}
\usepackage{fullpage}
\usepackage{amssymb}

\usepackage[symbol]{footmisc}

\usepackage{amsmath,amsthm,bbm}
\usepackage{amsfonts,amssymb}

\usepackage{hyperref}

\usepackage{graphicx}
\usepackage{epsfig}
\usepackage{times}
\usepackage{verbatim}
\usepackage{color}
\usepackage[normalem]{ulem}

\newcommand{\be}{\begin{equation}}
\newcommand{\ee}{\end{equation}}
\newcommand{\bq}{\begin{eqnarray}}
\newcommand{\eq}{\end{eqnarray}}
\newcommand{\bea}{\begin{eqnarray}}
\newcommand{\eea}{\end{eqnarray}}
\newcommand{\ba}{\begin{align}}
\newcommand{\ea}{\end{align}}

\newcommand{\e}{\text{e}}

\newcommand{\bR}{\mathbbm{R}}
\newcommand{\bN}{\mathbbm{N}}
\newcommand{\bC}{\mathbbm{C}}

\newcommand{\cV}{\mathcal{V}}

\newcommand{\cA}{\mathcal{A}}

\newcommand{\cW}{\mathcal W}

\newcommand{\cB}{\mathcal B}

\DeclareMathOperator{\arcos}{arcos}

\newtheorem{theorem}{Theorem}
\newtheorem{lemma}[theorem]{Lemma}

\newtheorem{proposition}[theorem]{Proposition}
\newtheorem{definition}[theorem]{Definition}

\newtheorem{remark}[theorem]{Remark}

\renewcommand{\>}{\rangle}
\newcommand{\<}{\langle}

\usepackage{tikz}
\usetikzlibrary{quantikz}
\usepackage{adjustbox}

\usepackage{authblk}

\title{Space-efficient Quantization Method for Reversible Markov Chains}

\author[1]{Chen-Fu Chiang}

\author[2,3]{Anirban Chowdhury}

\author[4]{Pawel Wocjan}

\affil[1]{Department of Computer and Information Science\\
State University of New York Polytechnic Institute\\
\texttt{chiangc@sunypoly.edu}}

\affil[2]{
Institute for Quantum Computing, University of Waterloo, Canada\\
\texttt{anirban.chnarayanchowdhury@uwaterloo.ca}}

\affil[3]{
Department of Combinatorics and Optimization, University of Waterloo, Canada}

\affil[4]{IBM Quantum\\ 
IBM T.J. Watson Research Center\\ 
Yorktown Heights, NY 10598, USA\\
\texttt{Pawel.Wocjan@ibm.com}}

\begin{document}

\maketitle


\abstract{In a seminal paper, Szegedy showed how to construct a quantum walk $W(P)$ for any reversible Markov chain $P$ such that its eigenvector with eigenphase $0$ is a quantum sample of the limiting distribution of the random walk and its eigenphase gap is quadratically larger than the spectral gap of $P$. 
The standard construction of Szegedy's quantum walk requires an ancilla register of Hilbert-space dimension equal to the size of the state space of the Markov chain. We show that it is possible to avoid this doubling of state space for certain Markov chains that employ a symmetric proposal probability and a subsequent accept/reject probability to sample from the Gibbs distribution. For such Markov chains, we give a quantization method which requires an ancilla register of dimension equal to only the number of different energy values, which is often significantly smaller than the size of the state space. 
To accomplish this, we develop a technique for block encoding Hadamard products of matrices which may be of wider interest.}


\section{Introduction}

Random walks and Markov chains have proven to be important paradigms in the design of provably efficient approximation algorithms in theoretical computer science and are also used extensively in computational physics. A discrete-time Markov chain is governed by a stochastic matrix $P$ whose entries $p_{yx}$ denote transition probabilities from a configuration $x$ to a configuration $y$ on a configuration space $\Omega=\{1,\ldots,N\}$. 
Assuming that the Markov chain is aperiodic and irreducible, it has a unique stationary distribution or \emph{fixed point} $\pi=(\pi_1,\ldots,\pi_N)^T$ with $\pi_x>0$ for all $x\in\Omega$
 and $P\pi=\pi$. The fixed point is also a limiting distribution, meaning that $\lim_{n\rightarrow \infty} P^n q=\pi$ for any initial probability distribution $q$. 

How fast a Markov chain approaches the limiting distribution $\pi$ is determined by the \emph{spectral gap} of its stochastic matrix $P$, usually denoted by $\Delta$, which is the difference between its two largest singular values. The spectral gap determines how many times one has to apply $P$ to any initial state to be able to sample from a distribution that is sufficiently close to $\pi$.

Szegedy developed a general method for quantizing classical algorithms based on reversible Markov chains \cite{szegedy04}. This led to several new quantum algorithms (see, for instance, \cite{arxiv,harrow2020adaptive} for algorithms for speeding up Markov chain Monte Carlo methods and \cite{apers2019unified} for quantum search, and the references therein).  

Szegedy's construction gives a quantum walk unitary $W(P)$ acting on $\bC^N \otimes \bC^N$ with the following important properties. Firstly, it has a unique eigenvector $|\psi\>$ of eigenvalue $0$ which is a coherent encoding of the stationary distribution of the Markov chain $P$. Moreover, all other eigenvectors of $W(P)$ have eigenvalues bounded away from 0 by a \emph{phase gap} $\phi$ which is quadratically bigger than the spectral gap $\Delta$ of $P$. 
The unique 0-eigenvalue eigenvector $|\psi\>$ has the form
\begin{align}
	|\psi\> 
	&= 
	\sum_{x\in\Omega} \sqrt{\pi_x} |x\> \otimes \sum_{y\in\Omega} \sqrt{p_{yx}} |y\> 
	=
	U( |\pi\> \otimes |0\> ),
\end{align}
where $U$ is a unitary we will define later, and $|\pi\>$ is the quantum sample 
\begin{align}
	|\pi\> = \sum_{x\in\Omega} \sqrt{\pi_x} |x\> 
\end{align}
of the limiting distribution $\pi$. 
Projecting onto, or reflecting about, the quantum sample $|\pi\>$ (or the state $|\psi\>$) is an important subroutine in several quantum algorithms including those for finding marked elements \cite{szegedy04, Krovi2015QuantumWC}, faster Markov chain mixing \cite{wocjan2008speedup, Orsucci2018fasterquantummixing} and simulations of classical annealing \cite{somma2008annealing}.

Quantum phase estimation of the walk unitary $W(P)$ can efficiently discriminate the $0$-eigenvector $|\psi\>$ from other eigenvectors and gives a method for projecting onto or reflecting about $|\psi\>$ requiring $O(1/\phi)$ uses of $W(P)$~\cite{Krovi2015QuantumWC}. The inverse dependence on the phase gap is key to polynomial speed-ups achieved by the quantum algorithms mentioned earlier. Subsequent work has also focused on reducing the gate and qubit overhead (i.e., on top of those needed to implement $W(P)$) required to implement such a reflection or projection~\cite{chowdhury2018improve,ge2019faster,gilyen2019quantum}. 

However, constructing the walk unitary $W(P)$ itself incurs significant overhead. In particular, Szegedy's construction requires an additional Hilbert-space with the same dimension as the configuration space of the Markov chain. Therefore, quantum circuit implementations of $W(P)$ require an ancillary register of $\Theta(\log N)$ qubits. In this work, we present a new quantization method that requires far fewer than $\log N$ qubits for a certain type of reversible Markov chains which we call \emph{propose-accept/reject Markov} chains. These Markov chains are ubiquitous and are used, for instance, in the famous Metropolis-Hastings algorithm \cite{hastings1970monte}. Efficient quantum walk implementations for the latter have been investigated previously \cite{Lemieux2020efficientquantum} but with the intent of reducing the quantum circuit depth. Our focus here instead is squarely on improving the ancilla overhead of Szegedy's quantum walk.

Our quantization method assumes that the Markov chain has the following property: there is an energy function $E : \Omega \rightarrow \{0,\ldots,B-1\}$ such that the limiting distribution $\pi$ is the corresponding Gibbs distribution at inverse temperature $\beta$. That is, its entries $\pi_x$ are given by $\pi_x = \frac{e^{-\beta E_x}}{Z_\beta}$, where the normalization factor $Z_\beta$ is the partition function. It is important to point out that in many situations $B$ is significantly smaller than $N$.  For instance, consider the $3$-SAT problem. Here the state space consists of all possible $n$-bit strings, that is, $N=2^n$.  The parameter $B$ is the maximum number of violated clauses, each violated clause contributing a penalty of $1$ to the overall energy of the truth assignment $x$.  In contrast to $N$, the number of clauses grows only polynomially as $O(n^3)$. In our construction, we only need an additional quantum register representing the subspace  $\bC^B$ instead of $\bC^N$.  
Thus, assuming this energy-dependence property, we quantize propose-accept/reject Markov chains with a unitary operator acting on $\bC^N \otimes \bC^B$. 

In what follows, we first briefly review classical Markov chains and Szegedy's quantization method in Section~\ref{sec:background}. Then, in Section~\ref{sec:propose} we describe propose-accept/reject Markov chains. Section~\ref{sec:quantization} provides our main result on a space-efficient quantization of propose-accept/reject Markov chains. We conclude in Section~\ref{sec:conclusions}.

\section{Background}\label{sec:background}

We begin by reviewing a few basic facts about Markov chains and refer the reader to \cite{yuval2017markov} for a detailed treatment.
Consider a Markov chain with state space $\Omega=\{1,\ldots,N\}$, where $N\in\bN$. It is represented by a column stochastic matrix $P=(p_{yx})\in\bR^{N\times N}$.  The (unique) stationary probability distribution of $P$ is the probability distribution $\pi=(\pi_x : x \in \Omega)^T\in\bR^N$ such that $P\pi = \pi$.  It is unique because we only consider ergodic Markov chains, which also implies that $\pi_x>0$ for all $x\in\Omega$.


\subsection{Reversible Markov chains}

We focus on the important case when the Markov chain is reversible. Then, the discriminant matrix $Q$ of $P$ is equal to $Q=D^{-1/2} P D^{1/2}$, where $D=\mathrm{diag}(\pi_1,\ldots,\pi_n)$.  It is symmetric, implying that its eigenvalues are all real. It is important to note that $P$ and $Q$ have the same spectrum since they are related by a similarity transformation.

The (two-sided) spectral gap $\Delta$ of the Markov chain is defined by
\begin{align}
	\Delta = 1 - \max\{\lambda_2,|\lambda_N|\},
\end{align}
where $\lambda_1=1 > \lambda_2 \ge \ldots \ge \lambda_N > -1$ are the  eigenvalues. The spectral gap determines how fast the Markov chain approaches the limiting distribution $\pi$.  In other words, it determines how many times one has to apply $P$ to any initial state to be able to sample from a distribution that is sufficiently close to $\pi$.

Observe that we can always make the eigenvalues non-negative by making the Markov chain lazy.  The transition matrix of the lazy Markov chain is simply $P'=\frac{1}{2}(I+P)$, and its spectral gap is equal to
$\Delta(P')=\frac{1}{2}(1-\lambda_2)$.There could be situations where it would be more beneficial to use the lazy version.\footnote{In the worst case, $\Delta(P')=\frac{1}{2} \Delta(P)$. However, $\Delta(P')>\Delta(P)$ when $\lambda_2 < 2|\lambda_N|-1$.}

The Markov chain is called \emph{reversible} if it satisfies
the \emph{detailed balanced condition}
\begin{equation}\label{eq:db}
	p_{xy} \pi_y = p_{yx} \pi_x
\end{equation}
for all $x,y\in\Omega$.

\begin{definition}[discriminant matrix]
	The discriminant matrix $Q$ of the reversible Markov chain $P$ is defined by
	\begin{equation}
		Q = D^{-\frac{1}{2}} P D^{\frac{1}{2}},
	\end{equation}
	where $D=\mathrm{diag}(\pi_x : x\in\Omega)$ is the diagonal matrix whose entries are the entries $\pi_x$ of the stationary distribution $\pi$.
\end{definition}

\begin{lemma}[Properties of discriminant matrix]\label{lem:discriminant}
	The discriminant matrix $Q=(q_{xy})$ of the reversible Markov chain $P=(p_{xy})$ is symmetric as its entries $q_{xy}$ are given by
	\begin{align}
		q_{xy} &= \sqrt{p_{xy}} \sqrt{p_{yx}}.
	\end{align}
\end{lemma}

\begin{proof}
	We have
	\begin{align}
		\sqrt{p_{xy}} \sqrt{p_{yx}} 
		&= 
		\sqrt{\pi_y^{-1} p_{yx} \pi_x} \, \sqrt{p_{yx}} \\
		&= 
		\sqrt{\pi_y^{-1}} p_{yx} \, \sqrt{\pi_x} \\
		&=
		q_{xy}, 
	\end{align}
	where we used the equality $p_{xy} = \pi_y^{-1} p_{yx} \pi_x$, which follows from the detailed balanced condition. 
\end{proof}

Observe that the symmetric matrix $Q$ and the stochastic matrix $P$ have the same spectrum. Thus, $Q$ has spectral gap $\Delta$. The quantum sample (qsample)
\begin{equation}
	|\pi\> = \sum_{x\in\Omega} \sqrt{\pi_x} |x\>
\end{equation}
of the stationary distribution $\pi$ is the eigenvector of the discriminant matrix $Q$ with eigenvalue $1$ as $|\pi\> = D^{-1/2}\pi$ and $Q=D^{-1/2} P D^{1/2}$.

\subsection{Szegedy's walk construction}\label{sec:szegedy_walk}
A simplified version of Szegedy's construction following Childs' lecture notes~\cite[17.2 How to quantize a Markov chain]{childs_qa} is as follows.
The discriminant matrix $Q$ is block encoded as 
\begin{align}
	Q &= T^\dagger \, S \, T,
\end{align} 
where $T$ is a certain isometry $T : \bC^N \rightarrow \bC^N \otimes \bC^N$ and $S$ denotes the swap operator $S$ acting on $\bC^N\otimes \bC^N$.  The isometry $T$ is directly related to the so-called quantum update $U$ which acts as
\begin{align}\label{eq:isometry-szegedy}
	U ( |x\> \otimes |0\> ) 
	&= 
	T|x\> \\
	&= 
	|x\> \otimes \sum_{y\in\Omega} \sqrt{p_{yx}} |y\>
\end{align}
for all $x\in\Omega$. 

The walk unitary $W(P)$ is given by
\begin{align}
	W(P) = S \, (2\Pi - I),
\end{align}
where $\Pi=T^\dagger T$ is the projector onto the image of $T$.  For the sake of completeness, we have included a detailed proof in Theorem \ref{thm:spec_q} of the well-known result showing that the walk unitary $W(P)$ ``stretches'' the spectrum of $Q$ onto the unit circle via the map
\begin{align}\label{eq:stretch}
	\lambda \mapsto e^{\pm i \arcos(\lambda)}.
\end{align}
This map quadratically amplifies the one-sided spectral gap $\Delta^+$, defined by 
\begin{align}
	\Delta^+ = 1 - \lambda_2,
\end{align}
to the \emph{phase gap} $\phi$ of $W(P)$, that is,  
\begin{align}
	\phi=\Theta(\sqrt{\Delta^+}),  
\end{align}  
where $\phi$ is the minimum non-zero eigenphase of the eigenvectors of $W(P)$. Observe that the one-sided gap $\Delta^+$ is bounded from below by $\Delta^+(P)\ge\Delta(P)$ and is equal to $\Delta^+(P)=2\Delta(P')$.  Therefore, one never needs to consider the lazy version of the Markov chain for quantum algorithms.\footnote{To the best of our knowledge, the existing literature on quantum walks does not explicitly mention that it is only the one-sided gap $\Delta^+$ that matters for quantum walks. An inspection of the map in (\ref{eq:stretch}) establishes this fact.}

\section{Propose-accept/reject Markov chains} \label{sec:propose}

Let $E : \Omega \rightarrow \bR$ be a function that assigns an energy $E_x$ to each state $x\in\Omega$. For an arbitrary given inverse temperature $\beta\in\bR^+$, we want to construct a reversible Markov chain with stochastic matrix $P=(p_{yx})\in\bR^{N\times N}$ such that its (unique) stationary distribution $\pi=(\pi_x)\in\bR^N$ satisfies
\begin{equation}
	\pi_x \, \propto \, e^{-\beta E_x}
\end{equation}
for all $x\in\Omega$. The \emph{partition function} $Z_\beta$ is the normalization factor 
\begin{equation}
	Z_\beta = \sum_{x\in\Omega} e^{-\beta E_x}
\end{equation}
so that the entries $\pi_x$ of the stationary distribution $\pi$ are given by
\begin{equation}
	\pi_x = \frac{e^{-\beta E_x}}{Z_\beta}
\end{equation}
for all $x\in\Omega$. This probability distribution is called the \emph{Gibbs distribution at the inverse temperature} $\beta$.

We mention that the obvious fact that any probability can be written as the Gibbs distribution for a suitably chosen energy function.  Given $\pi$, simply define the energy function $E$ by setting $E_x = -\ln \pi_x$ for all $x\in\Omega$.  Then, $\pi$ is the Gibbs distribution for this energy function $E$ at $\beta=1$.

The detailed balanced condition is equivalent to 
\begin{align}\label{eq:db_equiv}
	\frac{p_{yx}}{p_{xy}} 
	&=
	\frac{\pi_y}{\pi_x} =
	e^{-\beta\Delta_{yx}},
\end{align}
where we define $\Delta_{yx}$ to be the energy difference
\begin{align}
	\Delta_{yx} &= E_y - E_x.
\end{align}

The stochastic matrix $P$ with the desired properties can be constructed using the following propose-accept/reject approach.  Let $S=(s_{yx})\in\bR^{N\times N}$ be any symmetric stochastic matrix.  Assume that the current state is $x$.  A new candidate state $y$ is proposed with probability $s_{yx}$.  This proposed state $y$ is accepted with probability $a_{yx}$ and rejected with probability $1-a_{yx}$. (Observe that when $y=x$, then this choice is necessarily always accepted, that is, we have $a_{xx}=1$ for all $x$.)

The entries $p_{yx}$ of the corresponding transition matrix $P$ are as follows. Its off-diagonal and diagonal entries are
\begin{align}
	p_{yx} &= a_{yx} \cdot s_{yx} \label{eq:p_yx} \\
	p_{xx} &= 1 - \sum_{y\neq x} a_{yx} \cdot s_{yx}, \label{eq:p_xx}
\end{align}
respectively.  We have
\begin{align}
	\frac{p_{yx}}{p_{xy}} = \frac{a_{yx} \cdot s_{yx}}{a_{xy} \cdot s_{xy}} = \frac{a_{yx}}{a_{xy}}.
\end{align}
since the proposal probability is symmetric. Thus, to ensure that the detailed balanced condition in (\ref{eq:db_equiv}) holds, we must have
\begin{align}\label{eq:db_equiv_a}
	\frac{a_{yx}}{a_{xy}} 
	&= 
	e^{-\beta\Delta_{yx}}.
\end{align}
We will also assume that 
\begin{align}\label{eq:f_Delta_yx}
	a_{yx}&=f(\Delta_{yx}),
\end{align}
where $f : \mathbb{R} \rightarrow (0,1]$ is a function\footnote{Such function $f$ must satisfy the functional equation $f(x)=e^{-\beta x} f(-x)$.} such that the above condition holds. 

\begin{remark}[Metropolis and Glauber acceptance probabilities]
	There are two well known choices for the acceptance probability. 
	The first is the \emph{Metropolis} acceptance probability given by
	\begin{align}
		a^{\mathrm{(M)}}_{yx} = \min\big\{ 1, e^{-\beta\Delta_{yx}} \big\}.
	\end{align}
	The second is the \emph{Glauber} acceptance probability given by
	\begin{align}
		a^{\mathrm{(G)}}_{yx} &= \frac{e^{-\beta\Delta_{yx}}}{1 + e^{-\beta\Delta_{yx}}}.
	\end{align}
\end{remark}


Since the proposal matrix $S$ is both stochastic and symmetric, it is automatically doubly stochastic.  Any doubly stochastic matrix can be written as a convex combination of permutation matrices:
\begin{align}
	S &= \sum_{k=1}^\kappa w_k \cdot \Pi_k.
\end{align}
This follows from the Birkhoff–von Neumann theorem \cite{birkhoff1946tres}.  

We mention that since $S$ is not only doubly stochastic but also symmetric, the characterization in \cite{cruse1975note} shows that $S=\sum_k w_k \cdot \frac{1}{2}(\Pi + \Pi^T)$.  However, we will not need this more refined characterization.

We assume that $\kappa$ is small and all permutation matrices $\Pi_k$ can be implemented efficiently on a quantum computer.  For instance, this is the case if we consider random walks on Cayley graphs with generating sets that are closed under taking inverses (this condition is necessary for the graph to be undirected).  In the case of the $n$-bit Boolean cube, the corresponding symmetric stochastic matrix $S$ can be simply written as
\begin{align}
	S = \frac{1}{n}
	(
	X \otimes I \otimes \cdots \otimes I + 
	I \otimes X \otimes \cdots \otimes I +
	\ldots
	I \otimes I \otimes \cdots \otimes X
	),
\end{align}
where $X$ denotes the bit flip operator.


\begin{definition}[Propose-accept/reject Markov chains]\label{def:parMC}
	For $x\in\Omega$, let $E_x$ denote the energy of state $x$.  For $x,y\in\Omega$, let $\Delta_{yx}=E_y - E_x$ denote the energy difference of the states $y$ and $x$.
	
	A propose-accept/reject Markov chain consists of a symmetric proposal probability $s_{yx}$ and an acceptance probability $a_{yx}$, where $a_{yx}$ satisfies the condition in (\ref{eq:db_equiv_a}). The entries $s_{yx}$ and $a_{yx}$ are collected in the matrices $S$ and $A$, respectively.
	
	We say that the propose-accept/reject Markov chain is efficient if 
	\begin{itemize}
		\item $S$ can be written as a convex combination of permutation matrices
		\begin{align}
			S &= \sum_{k=1}^\kappa w_k \cdot \Pi_k
		\end{align}
		with $\kappa=O(\log N)$, 
		\item the weights $w_k$ are known, 
		\item the permutation matrices $\Pi_k$ and their inverses $\Pi_k^T$ can be implemented efficiently, 
		\item the energies $E_x$ are contained in the set $\{0,1, \ldots,B-1\}\subset\mathbb{N}$ with $B=O(\log N)$ and can be computed efficiently, and 
		\item the acceptance probabilities $a_{yx}$ are of the special form $a_{yx}=f(\Delta_{yx})$ as in (\ref{eq:f_Delta_yx}) with the additional condition that the function $f : \{-B+1,\ldots,B-1\} \rightarrow (0,1]$ can be efficiently computed.
	\end{itemize}
	Here ``efficiently computable'' means that there is a quantum circuit consisting of $O(\log N)$ elementary gates and enabling us to compute the desired functions.
\end{definition}


\section{Ancilla-efficient block encoding of discriminant matrix}\label{sec:quantization}

We will now present our main-results. Below, we use $\odot$ to denote the entrywise product of matrices, which is called the Hadamard or Schur product. We first present a decomposition of the discriminant matrix $Q$ that uses the Hadamard product as an essential operation. Then we show that it is possible to obtain efficient block encodings of the relavant Hadamard products because one of the factors can be suitably compressed.  

\begin{proposition}[Decomposition of transition matrix]
Let $S$ and $A$ denote the proposal and acceptance probability matrices of a propose-accept/reject Markov chain as in the definition above. Its stochastic matrix $P$ can be written as
\begin{align}
	P &=  A \odot S + R,
\end{align}
where 
\begin{align}
    R &= \sum_{k=1}^\kappa w_k \cdot \Pi_k^T \big((J - A) \odot \Pi_k \big) \label{eq:R}
\end{align}
is a convex sum of diagonal matrices.
\end{proposition}
\begin{proof}
    Observe that the diagonal elements of the stochastic matrix $P$ can be written as
    \begin{align}
        p_{xx} &= s_{xx} + \underbrace{\sum_y (1 - a_{yx}) \cdot s_{yx}}_{r_{xx}}\;, \label{eq:q_xx}
    \end{align}
    which is a slight modification of the expression in (\ref{eq:p_xx}) for $p_{xx}$ using the fact that $a_{xx}=1$.
    The Hadamard product $ A \odot S$ reproduces the off-diagonal elements $p_{yx}$ of $P$ which follows from Eqn.~\eqref{eq:p_yx}, and also contains $s_{xx}$ in the diagonal. Thus, all we need to show is that the sum $\sum_{k=1}^\kappa w_k \cdot \Pi_k^T \big((J - A) \odot \Pi_k \big)$ gives us $r_{xx}$.
    
    To this end, let $M=\sum_{u,v\in\Omega} m_{vu} \cdot |v\>\<u|$ be a matrix with arbitrary entries $m_{vu}$.  We have
	\begin{align}
		\Pi^T ( M \odot \Pi ) 
		& = 
		\left( \sum_v |v\>\<\pi(v)| \right) 
		\left( M \odot \sum_u |\pi(u)\>\<u| \right) \\
		&= 
		\left( \sum_v |v\>\<\pi(v)| \right) 
		\left( \sum_u m_{\pi(u),u} \cdot |\pi(u)\>\<u| \right) \\
		&= 
		\sum_x m_{\pi(x),x} \cdot |x\>\<x|
	\end{align}
	Using the above result, we obtain
	\begin{align}
		\sum_{k=1}^\kappa w_k \cdot \Pi_k^T \big( (J - A) \odot \Pi_k \big)
		&=
		\sum_x \left( 1 - \sum_{k=1}^\kappa w_k \cdot a_{\pi_k(x),x} \right) |x\>\<x| \\
		&=
		\sum_x \left( 1 - \sum_y a_{yx} \cdot s_{yx} \right) |x\>\<x| \\
		&=
		\sum_x r_{xx} \cdot |x\>\<x|
	\end{align}
	We used here
	\begin{align}
		\sum_{k=1}^\kappa w_k \cdot a_{\pi_k(x),x}
		&=
		\sum_y a_{yx} \cdot \sum_{k=1}^\kappa w_k \cdot \<y|\Pi_k|x\> 
		=
		\sum_y a_{yx} \cdot \<y|S|x\>.
	\end{align}
    The above discussion also shows that $R$ is a convex sum of diagonal matrices.
\end{proof}

\begin{theorem}[Decomposition of discriminant matrix]\label{thm:decomp_q}
	Let $S$ and $A$ denote the proposal and acceptance probability matrices of a propose/accept-reject Markov chain as in the definition above.   
	Its discriminant matrix $Q$ can be written as
	\begin{align}\label{eq:Q_decomposition}
		Q &= 
		(G \odot A) \odot S + R
	\end{align}
	where $R$ is as in Eqn.~\eqref{eq:R}, and the entries $g_{yx}$ of the matrix $G$ are given by
	\begin{align} \label{eqn:gyx}
		g_{yx} &= e^{\frac{1}{2}\beta \Delta_{yx}}
	\end{align}
	and $J$ is the all-ones matrix.  
	
	The matrix $G\odot A$ is symmetric and its entries are contained in $[0,1]$.  The entries of $J-A$ are contained in $[0,1]$.  
\end{theorem}


\begin{proof}
	The entries $q_{yx}$ of the discriminant matrix $Q=D^{-\frac{1}{2}} P D^{\frac{1}{2}}$ can be expressed as
	\begin{align} 
		q_{yx}  
		&= 
		\frac{1}{\sqrt{\pi_y}} \cdot p_{yx} \cdot \sqrt{\pi_x} \\
		&= 
		\sqrt{e^{-\beta (E_x - E_y)}} \cdot p_{yx} \\
		&=
		e^{\frac{1}{2} \beta \Delta_{yx}} \cdot p_{yx} \\
		&=
		g_{yx} \cdot p_{yx}
	\end{align}
	The above shows $Q=G \odot P$. We obtain the desired form in Eq.~\eqref{eq:Q_decomposition} since the Hadamard product is both associative and distributive, and $G\odot R = R$ holds. The latter follows because $R$ is diagonal and $G$ has ones on the diagonal.

	Let us now verify that $G\odot A$ is symmetric and its entries are contained in $[0,1]$.  The first condition is seen by
	\begin{align}
		g_{yx} \cdot a_{yx}
		&=
		e^{\frac{1}{2} \beta\Delta_{yx}} \cdot f(\Delta_{yx}) \\
		&=
		e^{\frac{1}{2} \beta\Delta_{yx}} \cdot
		e^{-\beta\Delta_{yx}} \cdot f(-\Delta_{yx}) \\
		&=
		e^{-\frac{1}{2} \beta\Delta_{yx}} \cdot
		f(-\Delta_{yx}) \\
		&=
		g_{xy} \cdot a_{xy}.
	\end{align}
	To prove the second property, we may assume that $\Delta_{yx}$ is negative thanks to the symmetric (first) property.
	But in this case it is clear that $g_{yx} \cdot a_{yx} = e^{\frac{1}{2}\beta\Delta_{yx}} \cdot f(\Delta_{yx})$ is less or equal to $1$ (recall that the range of the function $f$ is $[0,1]$).
\end{proof}


\begin{definition}[Block encoding]\cite{gilyen2019quantum}
	Suppose $L\in\bC^{N\times N}$ is an arbitrary matrix.  Let $\gamma\in\bR^+$ and $c\in\bN$. We say that the $(\log N + c)$-qubit unitary $V$ is a(n) (exact) $c$-ancilla block encoding of $L$ with scaling factor $\gamma$ 
	whenever 
	\begin{align}\label{eq:block_encoding}
		L = \gamma
		\left(\bra{0^{\otimes c}}\otimes I \right) V 
		\left(\ket{0^{\otimes c}}\otimes I \right) .
	\end{align}
	We use the tuple $(c,\gamma)$ to specify the parameters of the block encoding.   
\end{definition}

The parameter $c$ represents the number of additional ancilla qubits required for unitary $V$ to encode matrix $L$ by augmenting the space dimension from  $\bC^{N\times N}$ to $\bC^{2^{c} N \times 2^{c} N}.$ Since $L$ is a matrix of arbitrary norm and $V$ has norm $1$, the scaling factor $\gamma$ is required. 

Note that block-encodings can also be approximate, encoding a matrix only up to a small error, but this more general notion will not be necessary to explain our results. For the remainder of this paper, all block-encodings should be assumed exact unless otherwise specified.


\begin{lemma}[Block encoding of symmetric proposal matrix $S$]\label{lem:be_S}
	There is a block encoding of the symmetric proposal matrix $S=\sum_{k=1}^\kappa w_k \Pi_k$ with parameters $(2\log \kappa-1, 1)$.  The gate complexity of the block encoding is $O(\kappa)$ times the maximum of the gate complexities of the permutation matrices $\Pi_k$ for $k=1,\ldots,\kappa$.
\end{lemma}


\begin{proof}
	The block encoding of $S$ with the above parameters can be obtained with the help of linear combination of unitaries, which we briefly explain.  We need an ancilla register consisting of $\log\kappa$ qubits, a state-preparation unitary $C$ which acts as below:
	\begin{align}
		C|0\> = \sum_{k=0}^{\kappa-1} \sqrt{w_k} |k\>,
	\end{align}
	and the unitary $\Lambda$
	\begin{align}
		\Lambda = \sum_{k=0}^{\kappa-1} |k\>\<k| \otimes \Pi_k.
	\end{align}
	that applies the unitaries $\Pi_k$ controlled on the state of the ancilla register. 
	It can be verified that $(C^\dagger \otimes I)\Lambda (C \otimes I)$ is a block encoding with parameters $(\log\kappa,1)$.
	
	The state-preparation unitary $C$ can be implemented with gate complexity $O(\kappa)$\footnote{A general unitary would need $O(\kappa^2)$ elementary operators (see \cite{shende2006synthesis}).  However, is suffices here to implement any unitary as long as the first column is the desired qsample of the probability distribution $w_0,\ldots,w_{\kappa-1}$.  This tast is refered to as state synthesis.}.
	The controlled unitary $\Lambda$ can be implemented with cost $O(\kappa)$ times the maximum of gate complexities of the permutation matrices $\Pi_k$. The so-called SELECT construction achieving this linear scaling with $\kappa$ is presented in \cite[Lemma G.7]{childs2018simulation}.  It requires $\log\kappa-1$ additional ancilla qubits.\footnote{Note that it should be possible to use fewer ancilla qubits to implement $\Lambda$ at the cost increasing the gate complexity.}
\end{proof}

We now turn our attention to block encoding of Hadamard products.


\begin{lemma}[Block encoding of Hadamard product]
	Assume $L,M\in\bC^{N\times N}$ are arbitrary matrices having block encodings $V$ and $W$ with parameters $(c_L,\gamma_L)$ and $(c_M,\gamma_M)$, respectively.  
	Then, there exists a block encoding of the Hadamard product $L\odot M$ with parameters $(\log N + c_L + c_M, \gamma_L\cdot\gamma_M)$.   
\end{lemma}

\begin{proof}
	First, observe that $V\otimes W$ is a block encoding of the tensor product $L\otimes M$ with parameters $(c_L+c_M, \gamma_L\gamma_M)$.
	Second, observe that we can extract the Hadamard product $L\odot M$ from the tensor product $L\otimes M$ as a submatrix with the help of the isometry $T : \bC^N \rightarrow \bC^N \otimes \bC^N$ defined by $T|x\>=|x\>\otimes |x\>$.  Then, we have
	\begin{align}\label{eqn:shur_iso}
		T^\dagger (L \otimes M) T = L \odot M.
	\end{align}
	We obtain the desired block encoding of $L\odot M$ with parameters $(\log N + c_L + c_M, \gamma_L\cdot\gamma_M)$ by combining both facts.  Observe that the term $\log N$ occurs because the qubits on which $M$ acts have to be considered as ancilla qubits for the block encoding of the Hadamard product. 
\end{proof} 


\begin{remark}
	We cannot use the above block encoding construction because the number of ancillas for the block encoding is at least $\log N$. This construction would double the Hilbert space just as Szegedy's construction.
	This is why we have to use that the matrix $L$ can be compressed.
\end{remark}


\begin{definition}[Energy-dependent matrix]
	Let $E : \Omega \rightarrow \{0,\ldots,B-1\}$ be an energy function. We say that the matrix $L\in\bC^{N\times N}$ is energy-dependent 
	if its entries have the special form
	\begin{align}
		\<y|L|x\>
		&=
		g(E_y, E_x),
	\end{align}
	where $g : \{0\ldots,B-1\} \times \{0\ldots,B-1\} \rightarrow \bR$ is an arbitrary function. We define the corresponding compressed matrix $\widehat{L}\in\bC^{B\times B}$ by
	\begin{align}
		\widehat{L} = \sum_{E',E} g(E', E) \cdot |E'\>\<E|.
	\end{align}
\end{definition}


\begin{lemma}[Compressed block encoding of Hadamard product]\label{lem:be_compressed_hadamard}
	Assume $L\in\bC^{N\times N}$ is an energy-dependent matrix with compressed matrix $\hat{L}\in\bC^{B\times B}$.  Let $M\in\bC^{N\times N}$ be an arbitrary matrix.  Assume $\hat{V}$ is a block encoding of the compressed matrix $\hat{L}$ with parameters $(\hat{c}_L,\hat{\gamma}_L)$ and $W$ is a block encoding of $M$ with parameters $(c_M,\gamma_M)$.
	
	Then, there exists a compressed block encoding of the Hadamard product  $L\odot M$ with parameters $(\log B + \hat{c}_L + c_M, \hat{\gamma}_L\cdot\gamma_M)$.
\end{lemma}

\begin{proof}
	This follows from 
	\begin{align}
		L \odot M 
		&=
		T_{E}^\dagger (\widehat{L} \otimes M) T_E,
	\end{align}
	where $\widehat{L}$ is the compressed matrix corresponding to $L$ and the ``energy'' isometry $T_E : \bC^N \rightarrow\bC^B\otimes\bC^N$ is defined by
	\begin{align}
		T_E|x\> = |E_x\> \otimes |x\>  
	\end{align}
	for all $x\in\Omega$. 
	We have
		\begin{align}
			\<y|L \odot M|x\>
			&=
			\ell_{yx} \cdot m_{yx} \\
			&=
			g(E_y, E_x) \cdot m_{yx} \\
			&=
			\<E_y|\widehat{L}|E_x\> \cdot \<y|M|x\> \\
			&=
			\big( \<E_y| \otimes \<y| \big) 
			\big( \widehat{L} \otimes M \big)
			\big( |E_x\> \otimes |x\> \big) \\
			&=
			\<y|T_E^\dagger \big( \widehat{L} \otimes M \big) T_E|x\>,
		\end{align}
		for all $x,y\in\Omega$.  This implies the desired identity $L\odot M = T_E^\dagger (\widehat{L} \otimes M) T_E$. We can now use the similar arguments as in the block encoding of the ``uncompressed'' Hadamard product.
	\end{proof}
	
	It is important that the matrices $G \odot A$ and $J - A$ appearing in the decomposition of the discriminant matrix $Q$ in Theorem~\ref{thm:decomp_q} are energy-dependent.  Thus, we can use the compressed Hadamard product construction discussed above. For this construction to be efficient, we need that the operator norm of the compressed matrices $\widehat{G\odot A}$ and $\widehat{J-A}$ is small.
	
	To bound their operator norm, we proceed as follows.  Observe that the entries of $\widehat{G\odot A}$ and $\widehat{J-A}$ are contained in $[0,1]$ since this holds for the matrices $G\odot A$ and $J-A$.  
	We use that the operator norm $\|X\|$ of an arbitrary matrix $X=(x_{ij})$ is bounded from above by
	\begin{align}
		\|X\|\le\sqrt{\|X\|_1 \cdot \|X\|_\infty},
	\end{align}
	where $\|X\|_1=\max_j \sum_i |x_{ij}|$ and $\|X\|_\infty=\max_i \sum_j |x_{ij}|$.  This bound implies the operator norms of the compressed matrices $\widehat{G\odot A}$ and $\widehat{J-A}$ are at most $B$ since their $1$ and $\infty$ matrix norms are at most $B$.

	\begin{lemma}[Block encoding of the compressed matrices $\widehat{G\odot A}$ and $\widehat{J-A}$]\label{lem:be_compressed}
		There are block encodings of $\widehat{G\odot A}$ and $\widehat{J-A}$ with parameters $(1, B)$. They can be realized with $O(B^2)$ elementary gates. 
	\end{lemma}

	\begin{proof}
		Let $\hat{L}$ stand for either $\widehat{G\odot A}$ and $\widehat{J-A}$.  We can use the singular value decomposition to write $\hat{L}/B=V \Sigma W^\dagger$, where $\Sigma=\mathrm{diag}(\sigma_0,\ldots,\sigma_{B-1})$ is a diagonal matrix with $\sigma_b\in [0,1]$ for $b=0,\ldots,B-1$. 
		For $\sigma\in[0,1]$, define the unitary $R(\sigma)$ acting on an ancilla qubit by
		\begin{align}
			R(\sigma) 
			&= 
			\left(
			\begin{array}{cc}
				\sigma            &  \sqrt{1-\sigma^2} \\
				\sqrt{1-\sigma^2} &  -\sigma
			\end{array}
			\right).
		\end{align}
		The controlled unitary $\Lambda$ acting on $\bC^2 \otimes \bC^B$ defined by
		\begin{align}
			\Lambda 
			&= 
			(I \otimes V) \,
			\left(
			\sum_{b=0}^{B-1} R(\sigma_b) \otimes |b\>\<b|
			\right)
			\,
			(I \otimes W^\dagger)
		\end{align}
		is a block encoding of $\hat{L}/B$ requiring only one ancilla qubit. 
		The unitary $\Lambda$ acts on a quantum register of dimension $2B$ and can be implemented using $O(B^2)$ elementary gates \cite{shende2006synthesis}.
	\end{proof}
	

	Using the decomposition of the discriminant matrix, the construction for the compressed Hadamard product, and constructions of block encodings for sums and product of block-encoded matrices, we can obtain an ancilla efficient block encoding for the discriminant matrix $Q$.
	
	\begin{theorem}[Ancilla efficient block encoding of discriminant matrix]
		Assume that we have a propose-reject/accept Markov chain as in Definition~\ref{def:parMC}.  Then, there exists a $(2\log \kappa + \log B + 2, 4B)$ block encoding of the discriminate matrix $Q$ with a unitary $W$ that is a reflection.
		
	\end{theorem}

	
	\begin{proof}
		First, we construct a $(2\log \kappa + \log B + 1, 2B)$ block encoding based on the decomposition 
		\begin{align}\label{eq:decomp_q}
			Q 
			&= 
			T_E^\dagger \big( \widehat{G \odot A} \otimes S\big) T_E + 
			\sum_{k=1}^\kappa w_k \cdot \Pi_k^T \cdot 
			T_E^\dagger \big( \widehat{J - A} \otimes \Pi_k \big) T_E,
		\end{align}
		where $T_E$ denotes the energy isometry. This decomposition is established by using the decomposition of the discriminant matrix $Q$ in Theorem~\ref{thm:decomp_q} and the fact that the matrices $G\odot A$ and $J-A$ are energy-dependent.  
		
		We obtain a block encoding of $(G\odot A)\odot S$ with parameters $(2\log\kappa + \log B, B)$ by using Lemma~\ref{lem:be_compressed} and Lemma~\ref{lem:be_S} together with 
		Lemma~\ref{lem:be_compressed_hadamard}.
		
		For each $k=1,\ldots,\kappa$, we obtain a block encoding of $(J-A)\odot \Pi_k$ with parameters $(\log B + 1, B)$ by using Lemma~\ref{lem:be_compressed} and the fact that $\Pi_k$ are unitary together with 
		Lemma~\ref{lem:be_compressed_hadamard}.  The block encoding of $\Pi_k^T\big((J-A)\odot \Pi_k\big)$ is obtained by multiplication with the unitary $\Pi_k^T$ from the left, which does not change the parameters.  The weighted sum $\sum_k w_k \cdot \Pi_k^T\big((J-A)\odot \Pi_k\big)$ can be realized with the linear combination approach as for the case of $S$ in Lemma~\ref{lem:be_S}. The resulting block encoding has parameters $(2\log \kappa + \log B, B)$.
		
		We can now combine the two block encodings.  This adds $1$ to the number of ancilla qubits and multiplies the normalization factor by $2$. Let $T^\dagger U T$ denote the resulting block encoding of $Q$.  
		
		Finally, we convert this block encoding into one that uses a reflection.  Define the reflection 
		\begin{align}
			W = |0\>\<1| \otimes U + |1\>\<0| \otimes U^\dagger
		\end{align}
		and the isometry $T_+$ defined by $T_+|x\> = |+\> \otimes (T|x\>) $, where $|+\>=\frac{1}{\sqrt{2}}(|0\> + |1\>)$.  
		This modification adds one additional ancilla qubit and multiplies the normalization factor by $2$.  We now have $Q/2=T_+^\dagger W T_+$, where we used that $Q$ is hermitian.
	\end{proof}

	\section{Conclusions}\label{sec:conclusions}
	We have presented a novel quantization method for a certain propose-accept/reject Markov chain.  Our method requires that the symmetric proposal stochastic matrix $S$ can be written as a convex sum of a small number of permutation matrices. It requires a significantly smaller number of ancilla qubits than Szegedy's method to block-encode the discriminant matrix of a propose-accept/reject Markov chain. To accomplish this, we used a new technique for constructing block encodings of Hadamard products of matrices. 
	We hope that some of these techniques and ideas could prove useful in quantum computing beyond quantizing Markov chains.

	We remark that while our method is more efficient in terms of space required, the original construciton of Szegedy however allows quantizing propose-accept/reject Markov chains with a richer class of proposal stochastic matrices.  For instance, it is possible to quantize so-called quantum-enhanced propose-accept/reject Markov chains.
	
	Finally, we mention that to have a detailed comparison of our method and Szegedy's method, one must apply the latter to the corresponding propose-accept/reject Markov chain (see section~\ref{sec:szegedy_par} in the appendix.) Furthermore, one must consider how many ancillas (in addition to the ancillas required for the doubling of state space) are necessary to realize the isometry in (\ref{eq:szegedy_par_isometry}). For brevity, we have briefly described Szegedy's construction for an arbitrary stochastic matrix $P=(p_{yx})$ in subsection~\ref{sec:szegedy_walk} without considering the situation when its entries $p_{yx}$ are not directly available, but arise from the two-stage process: (i) propose and (ii) accept/reject.  We have worked out the how Szegedy's construction can be extended and applied to propose-accept/reject Markov chains in Appendix in subsection~\ref{sec:szegedy_par}.
	
	\section{Acknowledgments}
	AC was supported in part by the Army Research Office under Grant Number W911NF-20-1-0014.
	
	\newcommand{\etalchar}[1]{$^{#1}$}

	\section{Appendix: Szegedy block encoding of reversible Markov chains}
	
	This is the standard way of block encoding that doubles the state space.
	
	
	\subsection{Reversible Markov chains}
	
	We consider the situation, where a reversible Markov chain $P=(p_{yx})$ is directly given.  For $x\in\Omega$, define the vectors
	\begin{equation}
		|\psi_x\> = |x\> \otimes \sum_{y\in\Omega} \sqrt{p_{yx}} |y\> \in \bC^N\otimes \bC^N
	\end{equation}
	and the subspace
	\begin{equation}
		\cA = \mathrm{span}\{ |\psi_x\> : x \in \Omega\> \} \subset \bC^N\otimes \bC^N.
	\end{equation}
	Let 
	\begin{equation}
		\Pi = \sum_{x\in\Omega} |\psi_x\>\<\psi_x|
	\end{equation}
	be the projector mapping $\bC^N\otimes \bC^N$ onto $\cA$ and
	\begin{equation}\label{eq:isometry}
		T = \sum_{x\in\Omega} |\psi_x\>\<x|.
	\end{equation}
	be the isometry mapping $\bC^N$ to $\cA$.
	
	
	\begin{definition}[Szegedy quantum walk]
		Define one step of the quantum walk to be the unitary
		\begin{equation}
			W(P) = S (2\Pi - I),
		\end{equation}
		where 
		\begin{equation}
			S = \sum_{x,y\in\Omega} |y\>\<x| \otimes |x\>\<y|
		\end{equation}
		denotes the swap operator acting on $\bC^N\otimes \bC^N$.
	\end{definition}
	
	For the purification $|\pi\>$ of the stationary distribution $\pi$, define the vector
	\begin{equation}
		|\psi_\pi\> = T|\pi\>,
	\end{equation}
	where $|\pi\>$ is the qsample of the limiting distribution $\pi$.
	
	
	\begin{lemma}
		The vector $|\psi_\pi\>$ is an eigenvector of $W(P)$ with eigenvalue $1$.
	\end{lemma}
	
	\begin{proof}
		We have $(2\Pi-I)|\psi_\pi\>=|\psi_\pi\>$ because $|\psi_\pi\>\in\cA$. We have $S|\psi_\pi\>=|\psi_\pi\>$ because of the detailed balanced condition:
		\begin{align}
			|\psi_\pi\> 
			&= \sum_{x\in\Omega} \sqrt{\pi_x} |x\> \otimes \sum_{y\in\Omega} \sqrt{p_{yx}} |y\> \\
			&= \sum_{x,y\in\Sigma} \sqrt{\pi_x} \sqrt{p_{yx}} |x\> \otimes |y\> \\
			&= \sum_{x,y\in\Sigma} \sqrt{\pi_y} \sqrt{p_{xy}} S \Big( |y\> \otimes |x\> \Big) \\
			&= S \left( \sum_{y,x\in\Sigma} \sqrt{\pi_y} \sqrt{p_{xy}} |y\> \otimes |x\> \right) \\
			&= S |\psi_\pi\>
		\end{align}
	\end{proof}
	
	\begin{lemma}
		The isometry $T$ defined in (\ref{eq:isometry}) satisfies the following three properties:
		\begin{align}
			T T^\dagger &= \Pi \\
			T^\dagger T &= I_N \\
			T^\dagger S T &= Q,
		\end{align}
		where $Q$ denotes the discriminant matrix of the reversible Markov chain $P$.
	\end{lemma}
	
	
	\begin{proof}
		These properties are verified as follows:
		\begin{align}
			%
			%
			T T^\dagger 
			&= \sum_{x,y\in \Omega} |\psi_x\>\<x|y\>\<\psi_y|  
			= \sum_{x\in \Omega} |\psi_x\>\<\psi_x| 
			= \Pi \\
			\nonumber \\
			%
			%
			T^\dagger T 
			&= \sum_{x,y\in \Omega} |x\>\<\psi_x | \psi_y\>\<y|  
			= \sum_{x\in \Omega} |x\>\<x| 
			= I_N \\
			\nonumber \\
			%
			%
			T^\dagger S T 
			&= \sum_{x,v\in \Omega} |x\>\<\psi_x| S |\psi_v\>\<v| \\
			&= \sum_{x,y,v,w\in \Omega} \sqrt{p_{yx}} \sqrt{p_{wv}} |x\>\<x,y| S |v,w\>\<v| \\
			&= \sum_{x,y,v,w\in \Omega} \sqrt{p_{yx}} \sqrt{p_{wv}} |x\>\<x,y|w,v\>\<v| \\
			&= \sum_{x,y,v,w\in \Omega} \sqrt{p_{yx}} \sqrt{p_{xy}} |x\>\<y| \\
			&= Q
		\end{align}
		In the last step we used the form of the discriminant matrix $Q$ derived in Lemma~\ref{lem:discriminant}.
	\end{proof}
	
	
	\subsection{Propose-accept/reject Markov chains}\label{sec:szegedy_par}
	
	We discuss how to quantize propose-accept/reject Markov chains.
	Our construction is motivated by ideas in \cite{yung2012}.
	For $x\in\Omega$, define the vectors
	\begin{equation}\label{eq:szegedy_par_isometry}
		|\psi_x\> = 
		|x\> \otimes 
		\sum_{y \in \Omega} \sqrt{s_{yx}} |y\> \otimes
		\Big( \sqrt{a_{yx}} |0\> + \sqrt{1 - a_{yx}} |1\> \Big) \in \bC^N \otimes \bC^N \otimes \bC^2.
	\end{equation}
	Let the subspace $\cA$, the isometry $T$, and the projector $\Pi$ as in the previous subsubsection.
	
	Define the quantum walk operator 
	\begin{equation}
		W(P) = R (2\Pi - I),
	\end{equation}
	where $R$ is now the controlled swap operator acting on $\bC^N\otimes\bC^N\otimes \bC^2$ given by
	\begin{equation}
		R = S \otimes |0\>\<0| + I_N \otimes I_N \otimes |1\>\<1|.
	\end{equation}
	Here $S$ is the swap operator acting on the first two registers.
	
	
	\begin{lemma}[discriminant matrix of proposed accept-reject Markov chain]\label{lem:par_disc}
		The the diagonal and off-diagonal entries discriminant matrix $Q=(q_{yx})$ of the propose-accept/reject Markov chain are given by
		\begin{align}
			q_{xy} &= \sqrt{a_{yx}} \cdot \sqrt{a_{xy}} \cdot s_{yx} \\
			q_{xx} &= p_{xx} = 1 - \sum_{y\neq x} a_{yx} \cdot s_{yx}
		\end{align}
	\end{lemma}
	
	\begin{proof}
		This follows from $q_{yx}=\sqrt{\pi_x / \pi_y} \cdot p_{yx} = \sqrt{a_{xy}/a_{yx}} \cdot p_{yx}$ and using the expressions for the entries $p_{yx}$.
	\end{proof}
	
	
	\begin{lemma}[Block encoding of discriminant matrix]
		We have
		\begin{equation}
			T^\dagger R T = Q,
		\end{equation}
		where $Q$ is the discriminant matrix of the reversible Markov chain $P$ that is obtained by the Metropolis construction in the first section.
	\end{lemma}
	
	\begin{proof}
		We have
		\begin{align}
			&
			T^\dagger R T \\
			&=
			\sum_{x,v} |x\>\<\psi_x| R |\psi_v\>\<v| \\
			&=
			\sum_{x,y,v,w} 
			\sqrt{s_{yx}} \sqrt{s_{wv}} 
			\left[ \<x, y| \left( \sqrt{a_{yx}} \<0| + \sqrt{1-a_{yx}} \<1| \right) \right]
			R
			\left[ |v, w\>  \left( \sqrt{a_{wv}} |0\> + \sqrt{1-a_{wv}} |1\> \right) \right]
			|x\>\<v| \\
			&=
			\sum_{x,y,v,w} 
			\sqrt{s_{yx}} \sqrt{s_{wv}}  
			\left\{
			\sqrt{a_{yx}} \sqrt{a_{wv}} \, \<x,y|w,v\> + \sqrt{1-a_{yx}} \sqrt{1-a_{wv}} \, \<x,y|v,w\>
			\right\} 
			|x\>\<v| \\
			&=
			\sum_{x,y}
			\sqrt{s_{yx}} \sqrt{s_{xy}} \sqrt{a_{yx}} \sqrt{a_{xy}} \, |x\>\<y| +
			\sum_{x,y}
			\sqrt{s_{yx}} \sqrt{s_{yx}} \sqrt{1-a_{yx}} \sqrt{1-a_{yx}} \, |x\>\<x| \\
			&=
			\sum_x \sum_{y\neq x}
			s_{yx} \sqrt{a_{yx}} \sqrt{a_{xy}} \, |x\>\<y| +
			\sum_x \Big( s_{xx} + \sum_{y\neq x} s_{yx} (1-a_{yx})\Big) \, |x\>\<x|
		\end{align}
		In the last step we used the form of the discriminant matrix $Q$ derived in Lemma~\ref{lem:par_disc}.
	\end{proof}
	
	\subsection{Quantum enhanced propose-accept/reject Markov chains}
	
	Let $U=(\varphi_{yx})$ be an arbitrary symmetric unitary matrix.  Let $S=(s_{yx})$ be the symmetric stochastic matrix whose entries are given by 
	\begin{align}
		s_{yx} = |\varphi_{yx}|^2.
	\end{align}
	We refer to $S$ as the proposal stochastic matrix and to $U$ as the proposal unitary matrix.  We say that a propose-accept/reject Markov chain is quantum enhanced when its proposal stochastic matrix $S$ is obtain from a proposal unitary matrix. The work \cite{layden2022enhanced} provides both numerical and experimental evidence that quantum enhanced propose-accept/reject Markov chains could perform better in some situations.
	
	We show how to quantize quantum enhanced propose-accept/reject Markovs chains.  For $x\in\Omega$, define the vectors
	\begin{equation}\label{eq:szegedy_par_isometry}
		|\psi_x\> = 
		|x\> \otimes 
		\sum_{y \in \Omega} \varphi_{yx} |y\> \otimes
		\Big( \sqrt{a_{yx}} |0\> + \sqrt{1 - a_{yx}} |1\> \Big) \in \bC^N \otimes \bC^N \otimes \bC^2.
	\end{equation}
	Let $T$ and $S$ be defined as in the previous subsection.
	
	
	\begin{lemma}[Block encoding of discriminant matrix]
		We have
		\begin{equation}
			T^\dagger R T = Q,
		\end{equation}
		where $Q$ is the discriminant matrix of the quantum enhanced propose-accept/reject Markov chain with proposal unitary matrix $U$.
	\end{lemma}
	
	
	\begin{proof}
		We have
		\begin{align}
			&
			T^\dagger R T \\
			&=
			\sum_{x,v} |x\>\<\psi_x| R |\psi_v\>\<v| \\
			&=
			\sum_{x,y,v,w} 
			\bar{\varphi}_{yx} \cdot \varphi_{wv} 
			\left[ \<x, y| \left( \sqrt{a_{yx}} \<0| + \sqrt{1-a_{yx}} \<1| \right) \right]
			R
			\left[ |v, w\>  \left( \sqrt{a_{wv}} |0\> + \sqrt{1-a_{wv}} |1\> \right) \right]
			|x\>\<v| \\
			&=
			\sum_{x,y,v,w} 
			\bar{\varphi}_{yx} \cdot \varphi_{wv} 
			\left\{
			\sqrt{a_{yx}} \cdot \sqrt{a_{wv}} \, \<x,y|w,v\> + \sqrt{1-a_{yx}} \cdot \sqrt{1-a_{wv}} \, \<x,y|v,w\>
			\right\} 
			|x\>\<v| \\
			&=
			\sum_{x,y}
			\bar{\varphi}_{yx} \cdot \varphi_{xy} \cdot \sqrt{a_{yx}} \cdot \sqrt{a_{xy}} \, |x\>\<y| +
			\sum_{x,y}
			\bar{\varphi}_{yx} \cdot \varphi_{yx} \cdot \sqrt{1-a_{yx}} \cdot \sqrt{1-a_{yx}} \, |x\>\<x| \\
			&=
			\sum_{x,y}
			s_{yx} \cdot \sqrt{a_{yx}} \cdot \sqrt{a_{xy}} \, |x\>\<y| +
			\sum_{x,y}
			s_{yx} \cdot \sqrt{1-a_{yx}} \cdot \sqrt{1-a_{yx}} \, |x\>\<x| \\
			&=
			\sum_x \sum_{y\neq x}
			\sqrt{a_{yx}} \cdot \sqrt{a_{xy}} \cdot s_{yx} \, |x\>\<y| +
			\Big( 1 - \sum_{y\neq x} s_{yx} \cdot a_{yx} \Big) \, |x\>\<x|
		\end{align}
		In the last step we used the form of the discriminant matrix $Q$ derived in Lemma~\ref{lem:par_disc}.
	\end{proof}
	
	\begin{remark}
		The above derivation shows Szegedy's quantization method makes it possible to work with a richer class of proposal stochastic matrices $S$ than our quantization method.  Recall the our method requires that $S$ be expressed as a convex sum of a small number of permutation matrices.
	\end{remark}
 
	
	\section{Appendix: standard quadratic gap amplification}
	
	For the sake of completeness, we include a proof of the standard gap amplification method based on the product of two reflections.  Part of this section is based on the notes \cite[subsections 17.3 and 17.4]{childs_qa}.
	
	First, we prove that any hermitian matrix $Q\in\bC^{N\times N}$ with $\|Q\|_2\le 1$ can be written as $T S T^\dagger$, where $T : \bC^N\rightarrow \bC^{2N}$ is a simple isometry and $S\in\bC^{2N\times 2N}$ a simple reflection.  We explicitly construct a suitable $T$ by assuming access to the eigenbasis and eigenvalues of $Q$. Second, we describe how to quadratically amplify the gap $\Delta$ of $Q$, where the gap is defined to be the distance between the maximum eigenvalue $\lambda_1(Q)=1$ and the second largest eignenvalue $\lambda_2(Q) \le 1-\Delta$.
	Third, we discuss how to perform gap amplification when $T$ is any isometry and $S$ any reflection such that $T^\dagger S T=Q$.
	
	Let $|\varphi_j\>$ be an orthonormal basis of $\bC^N$ consisting of eigenvectors of $Q$ with eigenvalues $\lambda_j\in\bR$.  We sort the eigenvalue in non-increasing order, that is, $1=\lambda_1\ge \lambda_2\ge\ldots\ge \lambda_N > -1$.  We assume that the maximum eigenvalue $\lambda_1$ is equal to $1$ and is that the second largest eigenvalue is strictly smaller than $1$, that is, $\lambda_2=1-\Delta^+$ for some $\Delta^+\gneqq 0$.  We refer to $\Delta^+$ as the one-sided gap.  We also assume that the smallest eigenvalue $\lambda_N$ is strictly greater than $-1$.
	
	
	\begin{lemma}[Representation of hermitian matrices as submatrices of reflections]\label{lem:tst}
		Let $Q\in\bC^{N\times N}$ be an arbitrary hermitian matrix.  Let $|\varphi_j\>$ be an orthonormal basis of $\bC^N$ consisting of eigenvectors of $Q$ with eigenvalues $\lambda_j$.  Assume that the spectrum is contained in the interval $(-1, 1]$.  Then, there exists an isometry $T\in\bC^{2N\times N}$ and reflection $S\in\bC^{2N\times 2N}$ such that
		\begin{align}
			T^\dagger S T &= Q.
		\end{align}
	\end{lemma}
	
	
	\begin{proof}
		For $j\in\{1,\ldots,N\}$, choose the angles $\theta_j\in [0, \pi]$ such that $\cos(\theta_j)=\lambda_j\in (-1, 1]$, and define the states
		\begin{align}
			|\chi_j\> &= 
			|j\> \otimes \Big( \cos(\theta_j / 2) |0\> + \sin(\theta_j / 2) |1\> \Big) \in \bC^N \otimes \bC^2.
		\end{align}
		Define the isometry $T$ to be
		\begin{align}
			T 
			&= 
			\sum_{j=1}^N |\chi_j\>\<\varphi_j|
		\end{align}
		and the reflection $S$ to be
		\begin{align}
			S 
			&= 
			I \otimes Z,
		\end{align}
		where $Z=\mathrm{diag}(1, -1)$.
		
		
		Using the trigonometric identity $\cos^2 x - \sin^2 x = \cos(2x)$, we obtain
		\begin{align}
			\<\varphi_j|T^\dagger S T |\varphi_k\> 
			&= 
			\<\chi_j|S|\chi_k\> \\ 
			&= \delta_{jk} \cdot \Big( \cos^2(\theta_j/2) - \sin^2(\theta_j/2) \Big) \\
			&=
			\delta_{jk} \cdot \cos(\theta_j) \\
			&=
			\delta_{jk} \cdot \lambda_j \\
			&= \<\varphi_j|Q|\varphi_k\>
		\end{align}
		for all $j,k\in\{1,\ldots,M\}$, which is equivalent to $T^\dagger S T = Q$.
	\end{proof}
	
	\medskip
	\noindent
	We need to define some subspaces of $\bC^{2N}$. Let $\cA=\mathrm{span}\{|\chi_j\> : j\in\{1,\ldots,N\}\}$, $S\cA=\mathrm{span}\{S|\chi_j\> : j\in\{1,\ldots,N\}\}$, and $\cB=\cA+S\cA$. We will also consider $\cB^\perp$, which is the orthogonal complement of $\cB$.
	
	
	\begin{theorem}[Spectrum of product of two reflections]
		Let $U$ be the product of two reflections
		\begin{equation}
			U = S (2\Pi - I),
		\end{equation} 
		where the projector $\Pi$ is defined to be
		\begin{equation}
			\Pi = \sum_{j=1}^M |\chi_j\>\<\chi_j|.
		\end{equation}
		We assume that the maximum eigenvalue $\lambda_1=1$, the second largest eigenvalue $\lambda_2=1-\Delta^+$ for $\Delta^+\gneqq0$, and the minimum eigenvalue $\lambda_M \gneqq -1$.  The spectrum of $U$ is as follows:
		\begin{enumerate}
			\item for $j=1$, the state $|\psi_1\>=|\chi_1\>=|1\>\otimes |0\>\in\cB$ is an eigenvector of $U$ with eigenvalue $1$,
			\item for $j\in\{2,\ldots, M\}$, the orthogonal states 
			\begin{equation}
				|\psi_j^\pm\> = |\chi_j\> - \mu_j^\pm S |\chi_j\>\in\cB
			\end{equation}
			are eigenvectors with eigenvalues $\mu_j=\e^{\pm \theta_j}$, and 
			\item the state $|1\>\otimes |1\>\in\cB^\perp$ is an eigenvector with eigenvalue $1$.
		\end{enumerate}
	\end{theorem}
	
	
	\begin{proof}
		It is obvious that the mutually orthogonal subspaces $\cW_j = \mathrm{span}(|j\>\otimes|0\>, |j\>\otimes|1\>\}$ are invariant under both $(2\Pi - I)$ and $S$ for $j\in\{1,\ldots,N\}$.  
		
		Let us consider an arbitrary but fixed $j\in\{1,\ldots,N\}$. We make identify $\cW_j$ with $\bC^2$ and also write $\theta$ and $\lambda$ instead of $\theta_j$ and $\lambda_j$, respectively. Using this identification, the first reflection $R=2\Pi-I$ acts as 
		\begin{equation}
			\mathcal{R}=2|\chi\>\<\chi| - I,
		\end{equation}
		where
		\begin{equation}
			|\chi\> = \cos(\theta/2) |0\> + \sin(\theta/2) |1\>,
		\end{equation} 
		and the second reflection $S=I\otimes Z$ acts as 
		\begin{equation}
			\mathcal{S}=2|0\>\<0|-I.
		\end{equation} 
		Their product $\mathcal{S}\mathcal{R}$ is equal to the rotation
		\begin{equation}
			\left(
			\begin{array}{cc}
				\cos(\theta)  & \sin(\theta) \\
				-\sin(\theta) & \cos(\theta)
			\end{array}
			\right).
		\end{equation}
		The eigenvectors and eigenvalues are
		\begin{equation}
			\left(
			\begin{array}{c}
				1 \\
				\pm i
			\end{array}
			\right)
		\end{equation}
		and $e^{\pm i \theta}=e^{\pm i\arccos{(\lambda)}}$, respectively.
		
		Observe that the vectors
		\begin{equation}
			|\chi\> - e^{\pm i\theta} \mathcal{S} |\chi\> =
			\left(
			\begin{array}{c}
				\cos(\theta/2) (1 - e^{\pm i \theta})  \\
				\sin(\theta/2) (1 + e^{\pm i \theta})
			\end{array}
			\right)
		\end{equation}
		can also be chosen as eigenvectors.  To see this, we compute the ratio of the second and first entries
		\begin{equation}
			\frac{\sin(\theta/2)}{\cos(\theta/2)} \frac{1 + e^{\pm i\theta}}{1 - e^{\pm i\theta}} =
			\frac{\sin(\theta/2)}{\cos(\theta/2)} \frac{e^{\mp i\theta/2} + e^{\pm i\theta/2}}{e^{\mp \theta/2} - e^{\pm i\theta/2}} =
			\frac{\sin(\theta/2)}{\cos(\theta/2)} \frac{2\cos(\theta/2)}{\mp 2 i \sin(\theta/2)} = \pm i.
		\end{equation}
		The above analysis proves the statements in the second point. For the statements in the first and third points, observe that $\lambda_1$ and, that, $\theta_1=0$ and $|\chi_1\>=|0\> \otimes |0\>$.  In this particular case, it turns out that the two reflections $\mathcal{R}$ and $\mathcal{S}$ coincide so their product is the identity matrix.
	\end{proof}
	
	
	\begin{lemma}[Lower bound on quantum gap]
		Let $\Delta^+$ denote the classical gap of $Q$. The quantum gap of the corresponding unitary $U$ is $\theta=\arccos(1-\Delta^+)$ and is quadratically larger than the classical gap because 
		\begin{equation}
			\theta \ge \sqrt{2\Delta^+}.
		\end{equation}
	\end{lemma}
	
	
	\begin{proof}
		Write the second largest eigenvalue as $\lambda_2 = 1 - \Delta^+ = \cos(\theta)$ for a suitable $\theta$.  We have
		\begin{equation}
			1 - \Delta^+ = \cos(\theta) \ge 1 - \frac{\theta^2}{2},
		\end{equation}
		which implies the bound $\theta\ge\sqrt{2\Delta^+}$. 
	\end{proof}
	
	
	\begin{remark}
		We have shown that any hermitian matrix $Q$ can be written as $T^\dagger S T$ and that corresponding unitary $U$ has a quadratically larger gap.  We have constructed mathematically simple matrices for the isometry $T$ and the reflection $S$.  To achieve this, it suffices to work in the Hilbert space $\bC^N\otimes\bC^2$.  This is possible because the above  construction assumes that we have direct access to the eigenbasis and eigenvalues of $Q$.
		
		However, for quantum algorithms, we cannot rely on such mathematically simple matrices because all necessary transformations $T$ and $S$ have to be implemented by efficient quantum circuits. In particular, we need to work in the Hilbert space $\bC^M$ with $M\ge 2N$ to obtain suitable $T$ and $S$ to block encode $Q$.
	\end{remark}
	
	
	\begin{definition}[Quantization]
		Let $T\in\bC^{M\times N}$ an isometry from $\bC^N$ to $\bC^M$ with $\mathrm{im}(T)=\cA$\footnote[4]{The dimension of the subspace $\cA$ is necessarily equal to $N$.}, $\Pi\in\bC^{N\times N}$ an orthogonal projector with $\mathrm{im}(\Pi) = \cA$\footnote[5]{It is obvious that $T T^\dagger = \Pi$ and $T^\dagger T = I_N$.}, and $S\in\bC^{M\times M}$ a reflection. Let $Q\in\bC^{N\times N}$ be a hermitian matrix, $|\varphi_1\>,\ldots,|\varphi_N\>$ an orthonormal basis of $\bC^N$ consisting of eigenvectors of $Q$ with eigenvalues $\lambda_1=1\gneqq \lambda_2 = 1 - \Delta^+ \ge \lambda_3 \ge \ldots \ge \lambda_N \gneqq -1$.  We refer to $\Delta^+$ as the classical gap.
		
		We define the quantization of $Q$ to be
		\begin{align}
			U &= S(2\Pi - I)\in\bC^{M\times M}
		\end{align}
		provided that the isometry $T$ and reflection $S$ satisfy the condition
		\begin{equation}
			T^\dagger S T = Q.
		\end{equation}
		We refer to the value $\theta=\arccos(1-\Delta^+)$ as the quantum gap of the quantization $U$ of $Q$.
		
		For $j\in\{1,\ldots,N\}$, define the states
		\begin{equation}
			|\chi_j\>=T|\varphi_j\> \in \bC^M.
		\end{equation}
		Let $\mathcal{B}$ be the subspace $\mathcal{A}+S\mathcal{A}$, where $S\mathcal{A}=\{ S|\chi\> : |\chi\>\in \mathcal{A}\}$, and $\mathcal{B}^\perp$ the orthogonal complement of $\mathcal{B}$.
	\end{definition}
	
	\begin{lemma}
		The subspace spanned $|\chi_1\>,\ldots,|\chi_N\>$ coincides with the subspace $\cA$.  
	\end{lemma}
	
	
	\begin{proof}
		We have $\sum_{j\in [N]} |\chi_j\>\<\chi_j| = \sum_{j\in [N]} T |\varphi_j\>\<\varphi_j| T^\dagger = T T^\dagger = \Pi$. 
	\end{proof}
	
	
	\begin{theorem}[Spectrum of quantization]\label{thm:spec_q}
		The subspace $\mathcal{B}$ and its orthogonal complement $\mathcal{B}^\perp$ are invariant under $U$.  The spectrum of $U$ restricted to $\mathcal{B}$ is as follows:
		\begin{enumerate}
			\item For $j=1$, the one-dimensional subspace $\cV_1$ spanned by $|\chi_1\>$ is invariant under $U$ and the eigenvector of $U$ in $\cV_1$ is 
			\begin{equation}
				|\psi_1\> = |\chi_1\> \in \mathcal{B}
			\end{equation}
			with eigenvalue $1$.
			
			\item For $j\ge 2$, the two-dimensional subspace $\cV_j$ spanned by $|\chi_j\>$ and $S|\chi_j\>$ is invariant under $U$ and the two eigenvectors of $U$ in $\cV_j$ are
			\begin{equation}
				|\psi^{\pm}_j\> = |\chi_j\> - \mu^{\pm}_j S|\chi_j\> \in \mathcal{B}
			\end{equation}
			with corresponding eigenvalues $\mu^{\pm}_j$, where
			\begin{equation}
				\mu^{\pm}_j = \lambda_j \pm i \sqrt{1-\lambda_j^2} = e^{\pm i \, \mathrm{arccos}(\lambda_j)}.
			\end{equation}
		\end{enumerate}
	\end{theorem}
	
	\begin{proof}
		We use the properties of the isometry $T$ to proceed.  We begin with the case $j\ge 2$, that is, $\lambda_j$ is bounded away from $1$ by at least the spectral gap $\Delta^+$. First, we obtain 
		\begin{align}
			U |\chi_j\> 
			&= S (2\Pi - I)|\chi_j\> \\
			&= S (T T^\dagger - I) T |\varphi_j\> \\
			&= 2 S T |\varphi_j\> - ST|\varphi_j\> \\
			&= S |\chi_j\>.
		\end{align}
		Second, we obtain
		\begin{align}
			U S|\chi_j\> 
			&= S(2\Pi - I) S T |\varphi_j\> \\
			&= S(2 T T^\dagger - I) ST |\varphi_j\> \\
			&= (2S T T^\dagger S T - T) |\varphi_j\> \\
			&= (2S T Q - T) |\varphi_j\> \\
			&= (2\lambda_j S T - T) |\varphi_j\> \\
			&= (2\lambda_j S - I) |\chi_j\>.
		\end{align}
		We see that the subspace spanned by $|\chi_j\>$ and $S|\chi_j\>$ is invariant under $U$ so we can find eigenvectors of $U$ within this subspace. Let
		\begin{equation}
			|\psi^{\pm}_j\> = |\chi_j\> - \mu^{\pm}_j S|\chi_j\>.
		\end{equation}
		We have
		\begin{align}
			U|\psi^\pm_j\> 
			&= S|\chi_j\> - \mu^\pm_j (2\lambda_j S - I)|\chi_j\> \\
			&= \mu^\pm_j |\chi_j\>  - (2\lambda_j \mu^\pm_j - 1) S |\chi_j\>
		\end{align}
		Therefore, $|\psi^\pm_j\>$ is an eigenvector of $U$ with eigenvalue $\mu^\pm_j$ provided that 
		\begin{equation}
			(\mu^\pm_j)^2 - 2 \lambda_j\mu^\pm_j + 1 = 0,
		\end{equation}
		that is,
		\begin{equation}
			\mu^\pm_j = \lambda_j \pm i \sqrt{1 - \lambda_j^2} = e^{\pm i \, \mathrm{arccos}(\lambda_j)}.
		\end{equation}
		
		We now consider the case $j=1$.  We have
		\begin{equation}
			1 = \<\varphi_1|Q|\varphi_1\> = \<\varphi_1|T^\dagger S T|\varphi_1\> = \<\chi_1|S|\chi_1\>,
		\end{equation}
		implying that $|\chi_1\>$ is an eigenvector of $S$ with eigenvalue $1$.  We have
		\begin{equation}
			U|\chi_1\> = S(2\Pi - I)|\chi_1\> =  S |\chi_1\> = 1,
		\end{equation}
		where we used $\Pi|\chi\>=T T^\dagger T|\varphi_1\>= T|\varphi_1\>=|\chi_1\>$.
	\end{proof}
	
\end{document}